\nonstopmode
\documentclass[twoside,leqno,twocolumn]{article}
\usepackage{ltexpprt} 

\usepackage[utf8]{inputenc}
\usepackage[T1]{fontenc}

\usepackage{ccaption}
\captionnamefont{\bf}
\captiontitlefont{}
\captiondelim{. }

\renewcommand{\paragraph}[1]{\vspace{0.3em}{\bf\emph{#1.}}}

\usepackage{booktabs}
\usepackage{enumitem}

\usepackage{graphicx}

\usepackage{amsfonts,amssymb,latexsym}
\usepackage{better-proof-env}

\usepackage{mathtools}

\usepackage{dblfloatfix}

\newenvironment{remark}{\noindent\textbf{Remark.}}{}
\newenvironment{example}{\vspace{0.4em}\noindent\textit{Example.}}{}

\def\etal{\text{\textit{et~al.}}}

\makeatletter
\newskip\@bigflushglue \@bigflushglue = -100pt plus 1fil

\def\bigcentering{\let\\\@centercr\rightskip\@bigflushglue%
\leftskip\@bigflushglue
\parindent\z@\parfillskip\z@skip}

\makeatother

\usepackage{stmaryrd}

\newcommand{\set}[1]{\ensuremath{\left\{#1\right\}}}

\newcommand{\card}[1]{\ensuremath{{\left|{#1}\right|}}}

\newcommand{\prob}[2][]{\ensuremath{{\mathbb P}_{{#1}}\!\left[{#2}\right]}}
\newcommand{\expect}[2][]{\ensuremath{{\mathbb E}_{{#1}}\!\left[{#2}\right]}}

\def\smexp#1{\mathrm{e}^{#1}}
\def\lgexp#1{\ensuremath{\exp{\left({#1}\right)}}}

\def\N{{\ensuremath{{\mathbb N}}}}

\def\BernoulliLaw#1{\ensuremath{\mathrm{Ber}\!\left({#1}\right)}}

\def\GeometricLaw#1{\ensuremath{\mathrm{Geo}\!\left({#1}\right)}}
\def\PoissonLaw#1{\ensuremath{\mathrm{Poi}\!\left({#1}\right)}}

\usepackage{mathabx}
\usepackage{pifont}
\def\ASDeath{\mbox{\ding{61}}}

\usepackage{hyperref}

\usepackage{aoa}
\usepackage{skull}

\title{Analytic Samplers and the Combinatorial Rejection Method}

\author{%
  Olivier Bodini\thanks{LIPN, UMR 7030, Universit\'e Paris 13/CNRS,
    F-93430 Villetaneuse, France.}\and %
  J\'{e}r\'{e}mie Lumbroso\thanks{Department~of~Computer~Science,
    Princeton University, 35~Olden~Street, Princeton, NJ 08540, USA.}
  \and %
  Nicolas Rolin$\text{}^*$}

\date{}

\newcommand{\tASamp}[2][z,a]{\ensuremath{\cramped{\Gamma#2(\cramped{#1})}}}

\begin{document}

\maketitle

\begin{abstract}
  Boltzmann samplers, introduced by Duchon~\textit{et al.} in 2001, make
  it possible to uniformly draw approximate size objects from any class
  which can be specified through the symbolic method. This, through by
  evaluating the associated generating functions to obtain the correct
  branching probabilities.

  But these samplers require generating functions, in particular in the
  neighborhood of their sunglarity, which is a complex problem; they also
  require picking an appropriate tuning value to best control the size of
  generated objects. Although Pivoteau~\etal have brought a sweeping
  question to the first question, with the introduction of their Newton
  oracle, questions remain.


  By adapting the rejection method, a classical tool from the random, we
  show how to obtain a variant of the Boltzmann sampler framework, which
  is tolerant of approximation, even large ones. Our goal for this is
  twofold: this allows for exact sampling with approximate values; but
  this also allows much more flexibility in tuning samplers. For the class
  of simple trees, we will show how this could be used to more easily
  calibrate samplers.
\end{abstract}

\begin{figure*}
  \centering
  \includegraphics[scale=0.6]{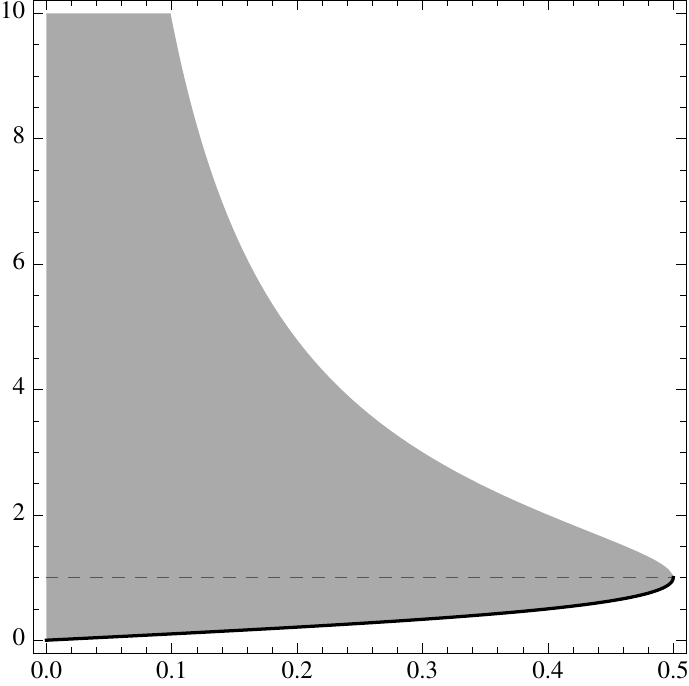}
  \caption{\label{fig:bidon}Plot associated with the combinatorial
    specification of binary trees where all nodes are counted,
    $\cls{B}=\clsAtom+\clsAtom\times \cramped{\cls{B}^2}$. The bottom
    thick black curve plots $C(z) = z+z\cramped{C(z)^2}$, or all
    coordinates $(z, C(z))$ usually considered for Boltzmann sampling; the
    shaded area is the region verifying $c \geqslant z+z\cramped{c^2}$,
    and from which we get coordinates $(z, c)$ which we use in our
    modified model.}
\end{figure*}

\section*{Introduction}

Being able to randomly generate large objects of any given combinatorial
class (for instance described by a grammar), is a fundamental problem with
countless applications in scientific modeling.

Nijenhuis and Wilf introduced the \emph{recursive method}~\cite{NiWi78} in
the late~70s (later extended by Flajolet~\etal~\cite{FlZiVa94}), the first
automatic random generation method; so termed \emph{automatic} because it
can directly derive random samplers from any combinatorial
description---no bijection, no clever algorithm, no complicated equations
are needed. The drawback is that this method is costly, notably in
preprocessing: to compute the probabilities involved in generating an
object of size $n$, the method requires knowing the complete enumeration
of the combinatorial class up to size $n$; and predictably when $n$ is
large, this enumeration is significant both to calculate and to store.

Enter \emph{Boltzmann sampling}, introduced by Duchon \emph{et al.} in
2002~\cite{DuFlLoSc02, DuFlLoSc04}, of which the key insight was that a
class' enumeration is not required to compute the correct branching
probabilities: instead, such probabilities can be obtained by evaluating
the counting generating functions---for an unlabelled combinatorial class
$\cls{C}$, for which there are $\cramped{c_n}$ elements of size $n$, its
counting generating function is defined as
\begin{align*}
  C(z) := \sum_{n=1}^\infty c_n z^n\text{.}
\end{align*}
Through evaluation, all the coefficients of a generating function are
smashed together, and the resulting probabilities take into account
objects of \emph{all sizes}. Thus, while you \emph{do} know that the
object returned will be \emph{uniformly} sampled among objects of the same
size, the size itself is a random variable---which you have no direct
control over. As a result, a significant aspect of Boltzmann sampling
involves: rejecting objects which are not within the desired size
interval; manipulating the generating functions so the size distribution
is such that not too many objects need be rejected.

The efficiency of this approach, combined with its mathematical
appeal---in many regards Boltzmann sampling is an elegant and natural
application of \emph{Analytic Combinatorics} pioneered by Flajolet and
Sedgewick~\cite{FlSe09}---have made it a fertile topic, and many of its
aspects have been developed through a broad number of papers.

\paragraph{The Boltzmann model} A Boltzmann sampler for an unlabelled
combinatorial class $\cls{C}$ (of which there are $\cramped{c_n}$ elements
of size $n$), is an algorithm that draws any given object
$\gamma\in\cls{C}$ with probability
\begin{align*}
  \prob[z]{\gamma} = \frac{z^{\card{\gamma}}}{C(z)}\qquad\text{with}\qquad
  C(z) := \sum_{n=1}^\infty c_n z^n\ =
  \sum_{\gamma\in\cls{C}} z^{\card{\gamma}}
\end{align*}
where $\card{\gamma}$ denotes the size of object $\gamma$ and $z$ is some
control parameter to be chosen. Thus the probability of obtaining an
object of size $n$ is
\begin{align*}
  \prob[z]{\card{\gamma}=n} = \frac{c_n z^n}{C(z)}\qquad\qquad
  \prob[s]{\gamma\ |\ \card{\gamma} = n} = \frac 1 {c_n}
\end{align*}
while the probability of drawing an object conditioned on its size is
uniform.

The name of the method is an analogy to the Boltzmann model of statistical
physics that assigns to each possible state of a system the probability
$\cramped{\smexp{-\beta E}/Z}$, where $E$ is the energy of the state,
$\beta=1/T$ is a constant, and $Z$ is normalizing constant---the original
authors noted that this was similar to the probability distribution of
objects. But in truth, the distribution of the sizes of objects is a very
generic distribution already known to probabilists as the \emph{Power
  Series Distribution}\footnote{The Poisson, geometric, log-series
  distributions are all special cases of this distributions, a fact which
  is put to use by Flajolet~\etal~\cite[\S 2]{FlPeSo11} who designed the
  Von Neumann/Flajolet scheme to simulate power series distributions using
  only random bits.}, and according to Johnson~\etal~\cite[\S
2.2]{JoKeKo05}, the terminology is usually credited to
Noack~\cite{Noack50} around 1950.

\paragraph{Evaluating GFs near their singularity}
Boltzmann samplers depend on the evaluation\footnote{A feature of
  Boltzmann samplers is that they generally require (see Otter trees in
  Section~\ref{sec:otter} for an exception) a constant number of such
  evaluations which can be charged a preprocessing; and this number of
  evaluation is dependent on the size of the combinatorial system, rather
  than on the size of the objects to be generated.} of generating function
in the neighborhood of their singularity\footnote{In keeping with the
  usage of analytic combinatorics, we call \emph{singularity}, the
  smallest positive point at which a generating function is not
  defined.}---and it was assumed that this could be done in constant time
arithmetic complexity.

The problem of how these functions should be evaluated was left open by
the original paper, and remained without answer until the contribution of
Pivoteau~\etal~\cite{PiSaSo08, PiSaSo12}. They introduced a variant of
Newton's iteration for combinatorial systems, which has a highly efficient
quadratic convergence (and what's more, is provably convergent for any
specifiable combinatorial class). However some aspects have remained open:

\begin{enumerate}[label=(\alph*)]
\item Solving the evaluation problem does not entirely solve the issue of
  tuning the samplers---that is, picking the value of $z$, which will
  yield the best concentration of objects of the targeted size. Currently,
  expected value tuning requires inverting a system of equations, and to
  our knowledge this is not routinely done for large combinatorial
  systems. For certain combinatorial classes (algebraic classes), singular
  samplers are tuned by approaching the singularity as close as possible
  using a binary search, requiring making a logarithmic number of calls to
  the oracle, as shown by Darrasse~\cite{Darrasse10}.

\item As a related issue, it seems legitimate to ask: is evaluating the
  generating functions truly necessary, since these evaluations are
  \emph{in fine} used to compute probabilities of average? and can
  relaxing this requirement possibly lead to more simple (or more
  efficient) implementations?

\item Finally, as a point of minor practical concern, but of conceptual
  interest: because of the finite nature of computers, and although the
  oracle can provide arbitrary precise values, the evaluation of
  generating functions is in practice restricted to fixed precision
  approximations. It has been argued that the incurred bias in uniformity
  is minimal: is it possible to make exact simulations from approximate
  values?
\end{enumerate}

\subsection*{Our contribution: an extended framework.}
This present paper attempts to investigate some of the aforementioned
questions. Our novel idea appeals to the classical random generation
concept of \emph{rejection}, as described for instance in Devroye's
chapter on the rejection method~\cite[\S 2]{Devroye86b}. Instead of
evaluating the generating functions exactly, we pick some nearby point
that is easier to compute.

This is illustrated in Figure~\ref{fig:bidon}. Both Boltzmann samplers and
our samplers use coordinates from the shaded region. But while Boltzmann
samplers limit themselves to the coordinates that belong to the thick
black curve at the bottom of the region, we allow ourselves to pick any
point within the region. Of course, this introduces a bias, which needs to
be compensated with some additional rejection. But we show that this
rejection is constant and that for reasonable choices of coordinates it is
practically negligible.

In this paper, we introduce this idea in Section~\ref{sec:main-concepts},
and showcase how it may be used through some illustrating examples: in
Section~\ref{sec:simple-trees}, we show our samplers can enable using
alternate techniques to determine the tuning parameter; in
Section~\ref{sec:otter}, we use the example of Otter trees (non plane
binary trees) to show how we can circumvent having to make a non-constant
number of evaluations of the generating function.

The ideas presented here followed from the first author's work to extend
Boltzmann samplers to infinite objects~\cite{BoMoTa12}, and the second
author's attempts to modify generating functions to shape the size
distribution of sampled objects.

\paragraph{Limits of this first version} This preliminary work comes with
a set of restrictions: because none of the authors are presently familiar
with the extensive litterature on multidimensional optimization, we have
avoided describing how to apply this idea to \emph{combinatorial systems}
(or multitype definitions), focusing instead on combinatorial classes
which can be described in a single equation\footnote{This was also
  the case of the original Boltzmann paper.}. This is not because the idea
of rejection cannot be trivially extended to systems, but rather because
we felt this strengthened our exposition, while the added complexity (in
notations, etc.) of describing systems could not be justified by our
current findings.

\section{Analytic Samplers\label{sec:main-concepts}}

In this section, we provide the main definitions for our analytic random
samplers, and then show the algorithms associated with the basic
constructions.

We have used the name \emph{Analytic Samplers} to distinguish our
contribution in this paper from Boltzmann samplers, but it should be noted
that the latter could legitimately be termed \emph{Analytic Samplers}.

\subsection{Main definitions}

\begin{Definition}
  Let $\cls{A}$ be an unlabelled combinatorial class, and $a_n$ the number
  of objects from $\cls{A}$ that have size $n$. The \emph{ordinary
    generating function} (OGF) associated with class $\cls{A}$ is defined
  equivalently by
  \begin{align*}
    A(z) := \sum_{n=0}^\infty a_n z^n\qquad\text{or}\qquad
    A(z) := \sum_{\alpha\in\cls{A}} z^{\card{\alpha}}\text{.}
  \end{align*}
\end{Definition}
The ordinary generating function enumerates combinatorial class it is
associated to. The tenet of the \emph{symbolic method}~\cite{FlSe09} is
that if a combinatorial class can be symbolically specified using a set of
operators (disjoint union, Cartesian product, sequence, multiset, etc.)
from initial terminal symbols called \emph{atoms} which have unit size,
then this specification can be directly translated to obtain the ordinary
generating function.

\begin{Definition}\label{def:phis}
  Let $\cls{A}$ be a symbolically defined combinatorial class which can be
  translated, following the symbolic method, to a functional
  equation on $A(z)$, the generating function associated with $\cls{A}$,
  \begin{align*}
    \cls{A} = \Phi(\clsAtom, \cls{A}, \mbox{\boldmath$\cls{X}$})
    \qquad\Rightarrow\qquad
    A(z) = \phi(z, A(z), {\mbox{\boldmath$X$}}(z))\text{,}
  \end{align*}
  where both $\Phi$ and $\phi$ may possibly involve other
  classes/generating functions which we note using vectors in bold (and
  each symbol/generating function component of the vector itself defined
  by their own equations).
\end{Definition}

\begin{Definition}
  Given a combinatorial class $\cls{A}$ as given in
  Definition~\ref{def:phis}, a pair of coordinates $(z,a)$ is said to be
  \emph{analytically valid coordinates} for the combinatorial class
  $\cls{A}$ if and only if they verify the inequality
  \begin{align*}
    a \geqslant \phi(z, a, \mbox{\boldmath$x$})\text{.}
  \end{align*}
\end{Definition}

\begin{remark}
  It is true that we could have some stronger bound, for instance, $a
  \geqslant A(z)$. But this is not desirable because the bound involves
  the generating function, the evaluation of which we are trying to avoid.
  
  Indeed, a subtle remark is that while computing the right-hand side of
  the inequality is not necessary to run the analytic samplers, it
  \emph{is} necessary to make the initial calibration. If this right-hand
  side depended on the generating function, we would be requiring strictly
  the same amount of work as traditional Boltzmann samplers---if not more!
\end{remark}

In general, for convenience and clarity, we will omit the vector
in the notations, and any additional bound symbols will be implicit.

\begin{Definition}\label{def:as}
  An \emph{analytic sampler} for an unlabelled combinatorial class
  $\cls{A}$ is an algorithm which samples an object $\alpha\in\cls{A}$, of
  size $\card{\alpha}$, with probability
  \begin{align*}
    \prob[(z,a)]{\alpha} = \frac {z^{\card{\alpha}}} {a}
  \end{align*}
  and fails with probability
  \begin{align*}
    \prob[(z,a)]{\ASDeath} = 1 - \frac{A(z)}{a}
  \end{align*}
  where $A(z)$ is the ordinary generating function associated with class
  $\cls{A}$, and the analytically valid coordinates $(z,a)$ are called the
  \emph{control parameter}. Moreover we denote by $\tASamp[z,a]{\cls{A}}$
  such an analytic sampler.
\end{Definition}

Because the original Boltzmann samplers already used the concept of
rejection to control the size of the output and constrain it to a
tolerance interval, we choose instead to call our additional rejection,
\emph{failure}, to avoid confusion.

\begin{theorem}
  Let $\cls{A}$ be a combinatorial class and $A(z)$ its generating
  function, and let $(z,a)$ be analytically valid coordinates for
  $\cls{A}$. The proportion of objects for which the generation has
  failed, does not depend on the size of the successfully generated
  object, and is equal to $1 - A(z)/a$.
\end{theorem}

\begin{proof}
  This follows from the definition of the model of Analytic Samplers
  wherein the probability of a single draw failing is constant---in the
  sense that it does not depend on the size of the object that was being
  constructed when the sampling failed---and equal to
  \begin{align*}
    \prob[(z,a)]{\ASDeath} = 1 - \frac{A(z)}{a}\text{,}
  \end{align*}
  and thus an object (of some random size) is drawn with the complementary
  probability. The number of failures before an actual object is drawn is
  then geometrically distributed with $p = 1 - A(z)/a$. We then have:
  \begin{align*}
    \expect[(z,a)]{\#\ASDeath}
      &= \sum_{k=0}^\infty k \left(1 - \frac{A(z)}{a}\right)^k \frac{A(z)}{a}\\
      &= \frac{A(z)}{a}
      \frac{\left(1 - \frac{A(z)}{a}\right)}{\left(\frac{A(z)}{a}\right)^2}\\
      &= \frac {a}{A(z)} \left(1 - \frac{A(z)}{a}\right) = \frac{a}{A(z)} - 1
  \end{align*}
  and because there is one last object generated (the one that does not
  fail, and after which we are done) the expected proportion of objects
  which have failed is
  $\expect[(z,a)]{\#\ASDeath}/(\expect[(z,a)]{\#\ASDeath} +1)$ as stated.
\end{proof}

Note that the lower bound for $a$ is $a = A(z)$, and for this choice of
value, the analytic sampler does not fail: the proportion of failed
objects is 0\%, and we revert to the case of Boltzmann samplers.

Indeed the inequality can naturally be seen as an equality involving a
slack variable $\delta$, $a = \phi(z, a, \mbox{\boldmath$x$}) + \delta$.
In essence, if you are willing to spend the computational time needed to
compute the generating function, then you are rewarded for your efforts by
having no rejection at all.

\subsection{Elementary constructions.}

In this subsection, we give the basic constructions used by our analytic
samplers. We follow the notation of the original
article~\cite{DuFlLoSc04}, and extend it to include our failure
probability,
\begin{align*}
  \Gamma\cls A : [p_1] \cdot \BernoulliLaw{p_2} \Rightarrow X\ |\ Y
\end{align*}
means that we first fail with probability $1-p_1$, then we draw a
Bernoulli variable $U$ of parameter $p_2$: if it is equal to $U=1$ then we
return $X$, otherwise we return $Y$. Furthermore, when instead of a
Bernoulli distribution, we have some variable $K$ drawn according to a
discrete distrbution, we mean that we return a tuple of $K$ independent
calls to the specified sampler.

Let $\cls{A}$, $\cls{B}$ and $\cls{C}$ be combinatorial classes. We
recall, for clarity, that in this article we note $\prob{\alpha}$ the
probability of drawing an object $\alpha$; when we want to make explicit
from which class this object is drawn, we note $\prob{\alpha\in\cls{A}}$.

\paragraph{Disjoint union} Let $\cls{A} = \cls{B} + \cls{C}$, and $a
\geqslant b + c$. We first reject $1-(b+c)/a$ of the objects, then we do a
normal Boltzmann.
\begin{align*}
  \Gamma\cls A: \left[\frac{b+c}{a}\right]\cdot
  \BernoulliLaw{\frac{b}{b+c}} \Rightarrow
  \Gamma\cls B \ \ |\ \ \Gamma\cls C
\end{align*}
\begin{proof}
  We must show that the sampler $\Gamma\cls{A}$ returns objects
  $\alpha\in\cls{A}$ with the correct probability $\prob[(z,a)]{\alpha} =
  z^{|\alpha|}/a$ (that is the probability of drawing an object from
  $\cls{A}$ follows the law given in Definition~\ref{def:as}), assuming
  inductively that the generators $\Gamma\cls{B}$ and $\Gamma\cls{C}$ are
  correct.

  Hence:
  \begin{align*}
    \begin{split}
      \prob[(z,a)]{\alpha\in\cls{A}} = \frac{b+c}{a}
      &\left(\frac{b}{b+c}\prob[(z,b)]{\alpha\in\cls{B}}\right.\\
        &\left. +\, \frac{c}{b+c}\prob[(z,c)]{\alpha\in\cls{C}}\right)
    \end{split}
  \end{align*}
  that is, the probability of sampling an object from $\cls{A}$ is the
  probability of first not failing, $(b+c)/a$, and then the probability of
  drawing the object using the sampler for class $\cls{B}$ or $\cls{C}$
  with the correct Bernoulli probability. By hypothesis those two samplers
  return objects with correct probabilities, so
  \begin{align*}
    \prob[(z,a)]{\alpha\in\cls{A}} &= \frac{b+c}{a}
    \left(\frac{b}{b+c}\frac{z^{|\alpha|}}{b}+
      \frac{c}{b+c}\frac{z^{|\alpha|}}{c}\right) =
    \frac{z^{|\alpha|}}{a}\text{.}
  \end{align*}
  Note that in this proof, and the following, we do not explicitly prove
  the probability of failure as it is a straightforward consequence: sum
  the probability of drawing an object over all possible possible objects,
  and take the complimentary probability.
\end{proof}

\paragraph{Cartesian product} Let $\cls{A} = \cls{B} \times \cls{C}$, and
$a \geqslant b \cdot c$.
\begin{align*}
  \Gamma\cls A: \left[\frac{b\cdot c}{a}\right]
   \Rightarrow
  \left(\Gamma\cls B \ \ ;\ \ \Gamma\cls C\right)
\end{align*}
\begin{proof}
  The proof follows the same model as the previous construction; let
  $\alpha=(\beta, \gamma)$,
  \begin{align*}
    \prob[(z,a)]{\alpha\in\cls{A}}&= \frac{b\cdot
      c}{a}\prob[(z,b)]{\beta\in\cls{B}}
    \prob[(z,c)]{\gamma\in\cls{C}}
  \end{align*}
  and since the samplers for $\Gamma\cls{B}$ and $\Gamma\cls{C}$ are
  inductively assumed to be correct,
  \begin{align*}
    \prob[(z,a)]{\alpha\in\cls{A}}&= \frac{b\cdot
      c}{a}\frac{z^{|\beta|}}{b}\frac{z^{|\gamma|}}{c}=
    \frac{z^{|\beta|+|\gamma|}}{a}=\frac{z^{|\alpha|}}{a}\text{.}
  \end{align*}
\end{proof}

\begin{example}
  At this point we will illustrate the initial definitions and these
  constructors by looking at the class $\cls{B}$ of binary trees, in which
  all nodes both internal and external count towards the size of the tree.
  These trees can be symbolically specified as either a leaf ($\clsAtom$)
  or a node which has two subtrees ($\clsAtom\times\cramped{\cls{B}^2}$),
  \begin{align}\label{eq:bintrees}
    \cls{B} = \clsAtom + \clsAtom\times\cls{B}^2\quad\Rightarrow\quad
    B(z)= z + zB(z)^2\text{.}
  \end{align}
  This functional equation can then be translated to an inequality,
  \begin{align*}
    b \geqslant z + zb^2\text{.}
  \end{align*}
  The analytically valid coordinates for $\cls{B}$ are all points that
  belong to the shaded region in Figure~\ref{fig:bidon}. The corresponding
  analytic sampler is:
  \begin{align*}
    \Gamma\cls B: \left[\frac{z+zb^2}{b}\right]\cdot
    \BernoulliLaw{\frac{z}{z+zb^2}} \Rightarrow
    \blacksquare \ \ |\ \ \left(\Gamma\cls B \ \ ;\ \ \Gamma\cls B\right)
  \end{align*}
  where $\blacksquare$ is a leaf of unit weight.
\end{example}

\paragraph{Sequence} Let $\cls{A} = \Seq{\cls B}$ and $a \geqslant
1/(1-b)$.
\begin{align*}
  \Gamma\cls A: \left[\frac{(1-b)^{-1}}{a}\right]\cdot
  \GeometricLaw{b} \Rightarrow
  \left(\Gamma\cls B,\ldots\right)
\end{align*}  

\begin{proof}
  We follow again the same model as previously. Let $\alpha=(\beta_1,
  \ldots,\beta_k)$.
  \begin{align*}
    \prob[(z,a)]{\alpha\in\cls{A}} =
    \frac{(1-b)^{-1}}{a}\prob{\GeometricLaw{b} = k}
    \prod_{i=1}^k \prob[(z,b)]{\beta_i\in\cls{B}}
  \end{align*}
  and by hypothesis
    \begin{align*}
      \prob[(z,a)]{\alpha\in\cls{A}} &=
      \frac{(1-b)^{-1}}{a}{b}^k(1-b)
      \prod_{i=1}^k \frac{z^{|\beta_i|}}{b}\\
      &= \frac{(1-b)^{-1}}{a}{b}^k(1-b)
      \frac{z^{|\beta_1|+\ldots+|\beta_k|}}{{b}^k}\\
      &= \frac{z^{|\beta_1|+\ldots+|\beta_k|}}{a} = \frac{z^{|\alpha|}}{a}\text{.}
  \end{align*}
\end{proof}

\subsection{Illustration of the rate of failure with Cayley Trees.}

For the purpose of giving an example that is somewhat more interesting
than binary trees, we go slightly beyond the scope of unlabelled
constructions which we have presented thus far, and present a
\emph{labelled} class which uses the \Set operator. The point of this
subsection is to illustrate, with an example, how moving away from the
curve of a generating function impacts the rate of failure---that is, the
rejection which must be done to compensate for sampling bias introduced by
the approximation.

Consider the example given by the class $\cls{T}$ of Cayley trees
(labelled, unrestricted, non-plane trees), symbolically specified as
$\cls{T} = \clsAtom \star \Set{\cls{T}}$. The exponential generating
function of this class, $T(z) = z\cramped{\smexp{T(z)}}$, is closely
related to Lambert's $W$-function, which is implicitly defined. Actual
standalone\footnote{By \emph{standalone}, we are referring to Boltzmann
  samplers not implemented within a computer algebra system, such as Maple
  or Mathematica, which usually provide computational access to such
  functions.} implementations of Boltzmann samplers requiring this
function have, for instance, have resorted to using its truncated Taylor
series expansion, see Bassino~\etal~\cite{BaDaNi07}.

\begin{figure*}
  \centering
  \raisebox{-0.5\height}{{
      \includegraphics[scale=0.7]{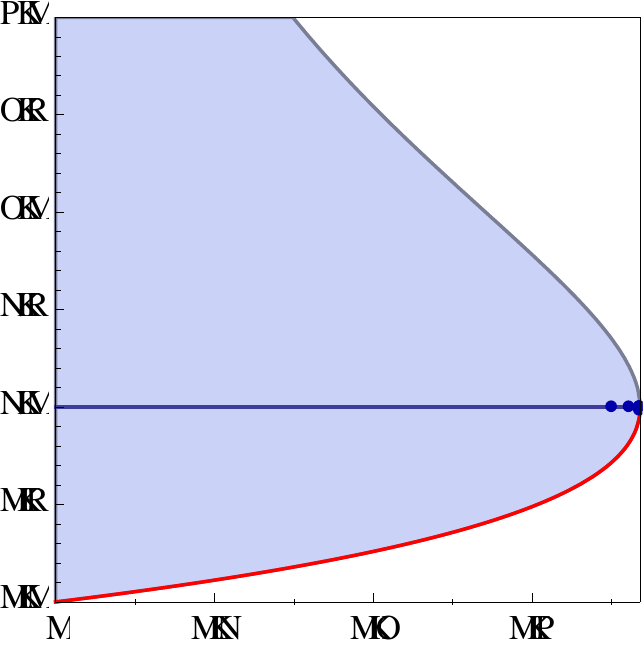}}}\qquad
  \raisebox{-0.5\height}{{
      \includegraphics[scale=1.0]{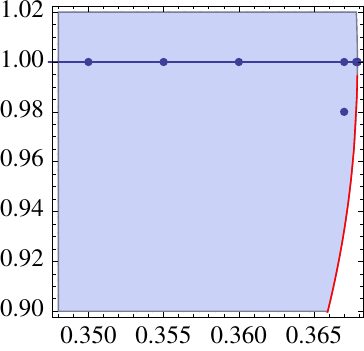}}}
  \caption{\label{fig:cayley-points}This is a plot of the region defined by
    the inequality $t \geqslant z\cdot\cramped{\lgexp{t}}$, with $z$ on
    the $x$-axis and $t$ on the $y$-axis. The lower bound of the region,
    in bolded-red, is the curve of the exponential generating function
    $T(z)$. Each point corresponds to one of the columns of
    Table~\ref{tbl:rejection}. The second figure, on the right, is a
    close-up near the singularity, at $z=\cramped{\smexp{-1}}$.}
\end{figure*}

With analytic samplers, our starting point is the system of functional
equations yielded by the symbolic method (here there is only a single
equation), replace any occurrence of a function by a free variable, and
obtain the inequality $t \geqslant z\cdot\cramped{\lgexp{t}}$ and the
algorithm
\begin{align*}
  \Gamma\cls{T}(z,t): \left[\frac{z\lgexp{t}}{t}\right]
  \cdot \PoissonLaw{t} \Rightarrow
  \square(\Gamma\cls{T}(z,t),\ldots,\Gamma\cls{T}(z,t))
\end{align*}
in other words, after an initial rejection (what we call failure) with
probability $z\cdot\cramped{\lgexp{t}}/t$ to account for the approximation
of the generating function, we draw a Poisson random variable of parameter
$t$, to indicate how many children to generate.

It is straightforward enough to see that this algorithm is correct for any
pair $(z,t)$ which satisfies the aforementioned inequality. Notice also
the failure ratio can easily be simulated exactly using techniques
described by Flajolet~\etal~\cite{FlPeSo11}.

An experiment summarized in Table~\ref{tbl:rejection} illustrates that the
impact of approximation is modest. For various pairs of $(z,t)$, the table
summarizes the result of making 1000 calls to the sampler: it indicates in
what proportion the sampling failed prematurely; and makes note of the
average and maximal size among the trees actually drawn. The case where
$z=\cramped{\smexp{-1}}$ and $t=1 = T(z)$ is special: first because this
is the only case in which $t$ is exactly equal to the evaluated EGF (thus
we have $0\%$ failure and our analytic sampler is a traditional Boltzmann
sampler); second because since we are evaluating the EGF in its
singularity, this is actual a \emph{singular Boltzmann sampler} (for which
the expected value of the size of the output in unbounded). All other
points, as illustrated in Figure~\ref{fig:cayley-points} are more or less
distant to the plot of $T(z)$, with a consequently higher failure rate:
but even at relatively significant distance from the curve, the failure
rate remains largely tolerable.

\begin{table*}
  \begin{bigcenter}
  \begin{tabular}{lcccccccc}
    \toprule
    \multicolumn{1}{l}{}&\multicolumn{7}{c}{$t=1$}&\multicolumn{1}{c}{$t=0.98$}\\
    \multicolumn{1}{r}{$z\;\;=$}& 0.35 & 0.36 & 0.367 & 0.3678 & 0.36787 & 0.367879 &$\smexp{-1}$ & 0.367\\
    \midrule
    failure (observed)& 28.8\% & 19.2\% & 6.4\% & 1.7\% & 0.4\% & 0.3\% &0\% & 3.5\%\\
    failure (theoretical)& 28.3\% & 19.4\% & 6.8\% & 2.1\% & 0.7\% & 0.2\% & 0\% & 3.9\%\\

    average size& 6.6 & 9.9 & 28.8 & 127. & 177.3 & 2716.7 &4944.3 & 35.9\\
    maximal size& 235 & 131 & 1493 & 17\,799 & 26\,531 & 826\,167 &2\,518\,975 & 1563\\
    \bottomrule
  \end{tabular}
  \end{bigcenter}
  \caption{\label{tbl:rejection}This table summarizes the result of making
    1000 calls to an analytic sampler for Cayley trees, with various
    values of $z$ (the control parameter) and $t$ (the approximation of
    the generating function). The \emph{failure} is the ratio of trees that
    must be rejected as a direct result of approximating the generating
    function, instead of evaluating it. Thus for the pair of values
    $z=\cramped{\smexp{-1}}$ and $t=1=T(z)$, in which we use the exact
    value of $T(z)$, our samplers are exactly equivalent to Boltzmann
    samplers, hence the failure is of $0\%$. What is remarkable is that the
    failure resulting from rather large approximations remains manageable.}
\end{table*}

\section{Simply generated trees\label{sec:simple-trees}}

We now would like to illustrate how dealing with a region (and inequality)
might make searching for an optimal pair of values for the sampler easier.
It this section, we show how we can search for the best value of the
tuning parameter $z$ (which happens to be in the vicinity of the
singularity) for a family of combinatorial classes, without evaluating the
generating function a logarithmic number of times.

Simply generated trees were introduced by Meir and Moon~\cite{MeMo78} as
classes of trees defined by the following specification
\begin{align}\label{eq:simpletrees}
  \cls{Y} = \clsAtom\times\Phi(\cls{Y})
\end{align}
where $\Phi$ is a polynomial defined as
\begin{align}\label{eq:simpletrees-phi}
  \Phi(x) = \sum_{\omega\in\Omega} x^\omega \qquad
  \Phi(x) = \sum_{\omega\in\Omega} \frac{x^{\left|\omega\right|}}{{\left|\omega\right|}!}
\end{align}
respectively depending on whether the class is unlabelled or labelled, and
where $\Omega \subseteq \N$ is the multiset of allowable degrees (for
instance, for binary trees, $\Omega = \set{0, 2}$). Meir and Moon
identified that trees families defined in such a way shared an important
number of common properties (such as mean path length of order $n\sqrt{n}$
or average height of order $\sqrt{n}$).

\subsection{Existing approaches.}

Randomly sampling from this class of tree is no longer particularly
challenging: there are several methods to do this, with various properties
of optimality (time, random-bit, etc.). So we do not presume to introduce
samplers with any sort of new efficiency. However the example of simply
generated trees illustrates a way in which calibration might be more
practical with analytic samplers.

\begin{figure*}
  \centering
  \includegraphics[scale=0.4]{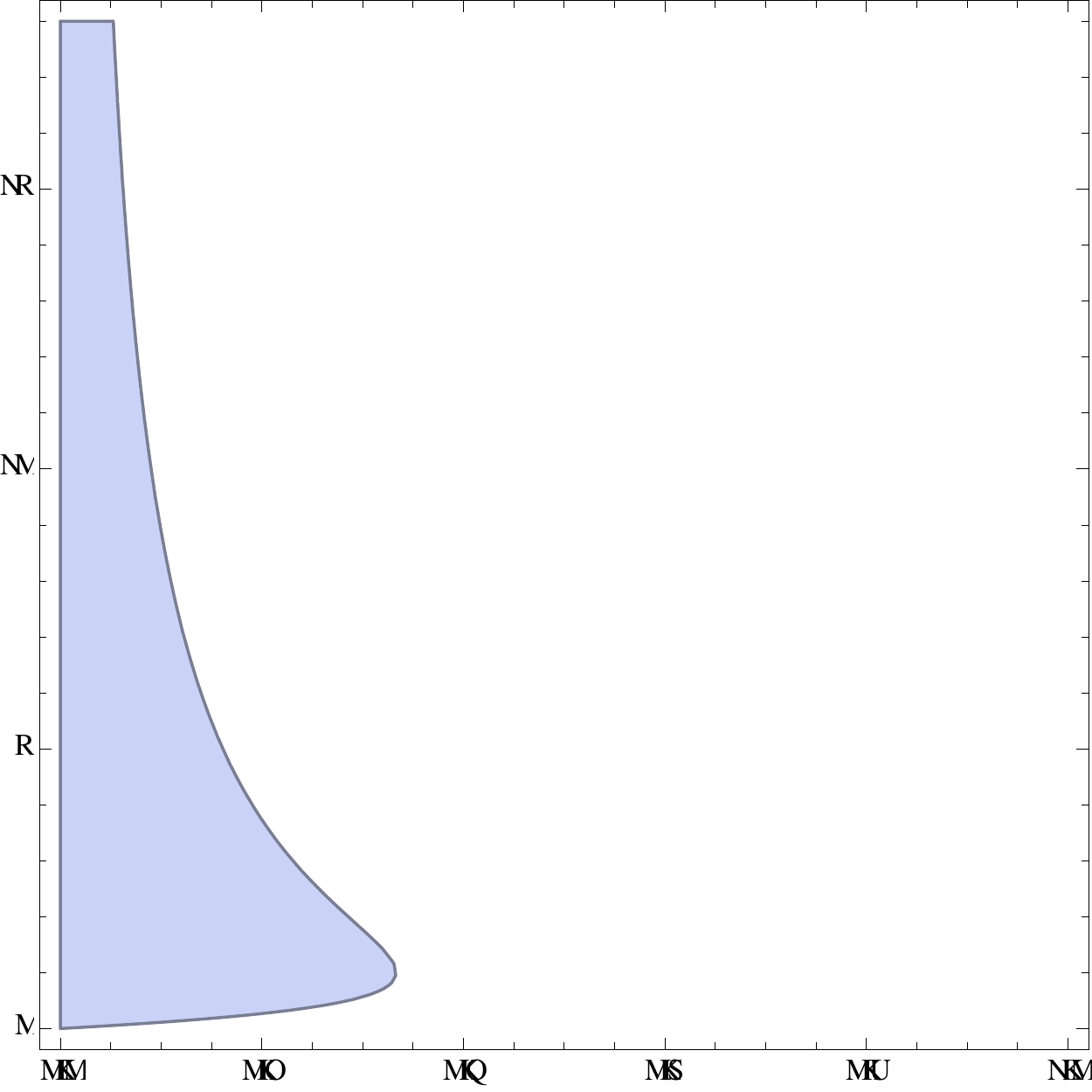}
  \includegraphics[scale=0.4]{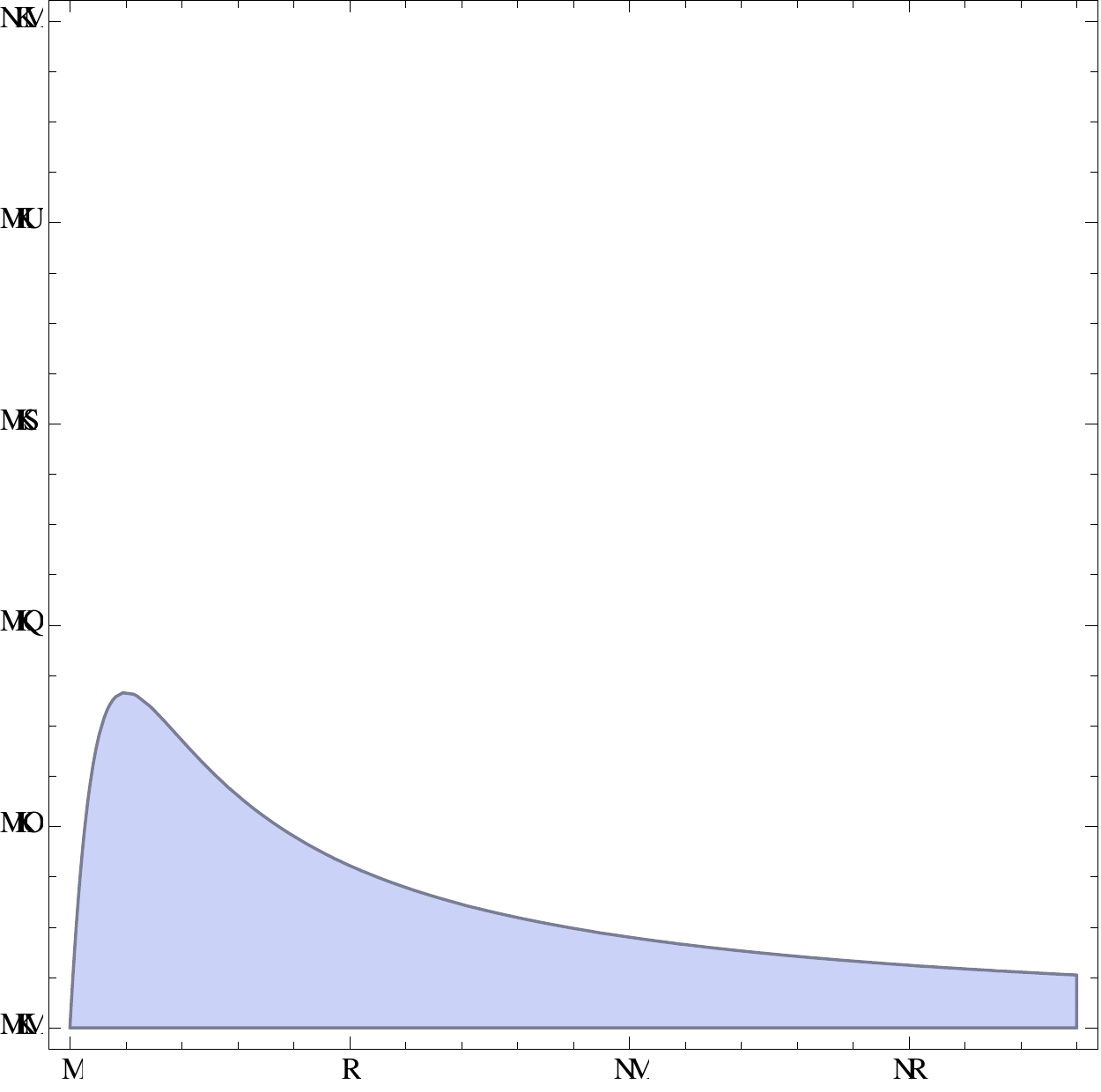}
  \caption{\label{fig:simple-trees}This is the plot of a the generating
    function of a simply generated tree (in this case, the class $\cls{U}$
    of unary-binary trees, i.e., $\Omega = \set{0, 1, 2}$). On the left,
    the $x$-axis is $z$, the control parameter and the $y$-axis has $U(z)
    = z\Phi(U(z))$. On the right, we are plotting the function $u\mapsto
    u/\Phi(u)$. The problem of looking for the singularity, left, has been
    reduced to the more palatable problem of maximizing a function,
    right.}
\end{figure*}

Simply generated trees happen to have a branching singularity. This
means: that their generating function can be evaluated at the singularity,
and also that the size distribution of objects produced by a Boltzmann
sampler would be `peaked', that is, highly concentrated towards smaller
objects. The solution has traditionally been to do singular sampling: to
pick $z$ as being at, or near, the singularity, generate objects with
unbounded expected size, and reject those that are too big.

Except in simple cases (such as binary trees, for which the singularity is
well known to be $1/4$), the singularity is \emph{not} known, so it must
be determined empirically. This is usually done with a binary search, as
implemented by Darrasse~\cite{Darrasse10}: the oracle introduced by
Pivoteau~\etal~\cite{PiSaSo12} converges when inside the radius of
convergence, and diverges otherwise; thus it is possible to detect whether
we have gone beyond the singularity. This method requires a logarithmic
number of calls to the oracle---a logarithmic number of evaluations that
are not done in constant time.

\subsection{New approach: maximize a polynomial.}

From the specification in Equation~\ref{eq:simpletrees}, we obtain the
condition for analytic-validity of a pair $(z,y)$,
\begin{align*}
  y \geqslant z\cdot\Phi(y)\text{.}
\end{align*}
With this, it is now easier, instead of looking at the generating function
$Y(z)$, to look instead at $y \mapsto y/\Phi(y)$, which is a rational
function. This function admits a maximal point in the unit interval, which
is the singular point of $Y(z)$.

Looking for this maximal point is a considerably easier problem, that does
not require any evaluation of the generating function (except perhaps for
an initial guess): it can be solved by differentiation, by Newton
iteration, or with specifically optimized algorithms available in the
litterature, such as Brent's algorithm~\cite{Brent73}.

\section{Substitution operator\label{sec:otter}}

While we have only described, for space and pertinence purposes, how to
build analytic samplers for classes using elementary constructors, the
possibilities are much broader. In particular, functional operators, such
as pointing (differentiation) or substituting (composition) can naturally
be used.

We will not go into detail, but instead provide the example of the
unordered pair, $\cramped{\MSet[2]}$, and present an application with the
random sampling of Otter trees.

\subsection{Unordered pair construction}

Let $\cls{B}$ be a combinatorial class, and $\cls{A}=\MSet[2]{\cls{B}}$ be
the class containing unordered pairs of elements of $\cls{B}$. The
corresponding generating functions $A(z)$ and $B(z)$ verify the functional
equation
\begin{align*}
  A(z) = \frac{B(z)^2 + B(\cramped{z^2})}{2}\text{.}
\end{align*}
Assuming there is an analytic sampler for $\cls{B}$, we can build an
analytic sampler for $\cls{A}$. Let $(z,b)$ and $(\cramped{z^2},\bar{b})$
both be analytically valid for $\cls{B}$ (note that the variable $z$ must
be the same in both pairs), and let $(z,a)$ be analytically valid for
$\cls{A}$, that is
\begin{align*}
  a \geqslant \frac{b^2+\bar{b}}{2}\text{.}
\end{align*}
Using the notation we have introduced,
\begin{align*}
\begin{split}
  \tASamp[z,a]{\cls{A}} : \left[\frac{b^2 + \bar{b}}{2a}\right] \cdot\ &
  \BernoulliLaw{\frac{b^2}{b^2+\bar{b}}} \Rightarrow\\
  &\left\{\tASamp[z,b]{\cls{B}}\ ;\ \tASamp[z,b]{\cls{B}}\right\}\ \ |\\
  &\tASamp[z^2,\bar{b}]{\cls{B}}\text{ and duplicate}\text{.}
\end{split}
\end{align*}
In other terms, after making the obligatory failure test, we choose with
the proper probability whether to create a pair of elements resulting from
independent calls to $\tASamp[z,b]{\cls{B}}$, or whether to make one call
to $\tASamp[\cramped{z^2},\bar{b}]{\cls{B}}$ and duplicating the resulting
object to make a pair of identical objects.

\begin{proof}
  As before, proving the validity of this algorithm involves showing that
  the analytic sampler $\tASamp[z,a]{\cls{A}}$ returns any object
  $\alpha\in\cls{A}$ with probability $\cramped{z^{\card{\alpha}}}/a$. We
  distinguish two disjoint cases.
  \begin{itemize}
  \item Either the pair $\alpha=\set{\cramped{\beta_1};
      \cramped{\beta_2}}$ contains two distinct elements,
    $\cramped{\beta_1}\not=\cramped{\beta_2}$. Then this pair could only
    have been produced by two independent (and distinguished) calls to
    $\tASamp[z,b]{\cls{B}}$. Thus under this setting,
    \begin{align*}
      \begin{split}
        \prob[z,a]{\alpha\ |\ \beta_1 \not=\beta_2} = \frac{b^2+\bar{b}}{2a} &\cdot
        \frac{b^2}{b^2+\bar{b}} \cdot\\
        &\left(\prob[z,b]{\beta_1}\cdot\prob[z,b]{\beta_2} \right.\\
        &\left.+\, \prob[z,b]{\beta_2}\cdot\prob[z,b]{\beta_1}\right)\text{.}
      \end{split}
    \end{align*}
    By hypothesis, $\tASamp[z,b]{\cls{B}}$ is an analytic sampler for class
    $\cls{B}$, which means it returns an object $\beta\in\cls{B}$ with
    probability $\cramped{z^{\card{\beta}}}/b$,
    \begin{align*}
        \prob[z,a]{\alpha\ |\ \beta_1 \not=\beta_2} &= \frac{b^2+\bar{b}}{2a} \cdot
        \frac{b^2}{b^2+\bar{b}}
        \frac{2z^{\card{\beta_1}+\card{\beta_2}}}{b^2} \\
        &= \frac{z^{\card{\beta_1}+\card{\beta_2}}}{a} = \frac{z^{\card{\alpha}}}{a}\text{.}
    \end{align*}
  \item Or the pair $\alpha = \set{\beta; \beta}$ contains two identical
    objects. The pair could then have been drawn by either branch: from
    two independent calls to $\tASamp[z,b]{\cls{B}}$ which happen to
    return the same object; or from the call to
    $\tASamp[z^2,\bar{b}]{\cls{B}}$ which is duplicated. In this case,
    \begin{align*}
      \begin{split}
        \prob[z,a]{\alpha\ |\ \beta =\beta} = \frac{b^2+\bar{b}}{2a} \cdot &\left(
          \frac{b^2}{b^2+\bar{b}}\cdot \prob[z,b]{\beta}^2 \right.\\
          &\left.+\,\frac{b^2}{b^2+\bar{b}}\cdot
          \prob[\cramped{z^2},\bar{b}]{\beta}\right)\text{.}
      \end{split}
    \end{align*}
    Assuming the analytic sampler for $\cls{B}$ is correct,
    \begin{align*}
      \begin{split}
        \prob[z,a]{\alpha\ |\ \beta =\beta} &= \frac{b^2+\bar{b}}{2a} \cdot
        \left(
          \frac{b^2}{b^2+\bar{b}}\cdot
          \left(\frac{z^{\card{\beta}}}{b}\right)^2\right.\\
        &\phantom{= \frac{b^2+\bar{b}}{2a} \cdot\quad}\left. +\,
          \frac{b^2}{b^2+\bar{b}}\cdot
          \frac{(\cramped{z^2})^{\card{\beta}}}{\bar{b}}\right)
      \end{split}
    \end{align*}
    which finally yields
    \begin{align*}
    \prob[z,a]{\alpha\ |\ \beta =\beta} &= \frac{z^{2\card{\beta}}}{a}
    = \frac{z^{\card{\alpha}}}{a}\text{.}
    \end{align*}
  \end{itemize}
\end{proof}


\subsection{Otter Trees}

We've already thoroughly discussed the class $\cls{B}$ of binary trees.
These binary trees are \emph{plane}, in the sense that there the children
of an internal node are distinguished: there is a left node and a right
node. We now consider the class $\cls{V}$ of \emph{Otter tree}, which are
binary trees that are \emph{non plane}, using the $\cramped{\MSet[2]}$
operator introduced in the previous subsection,
\begin{align*}
  \cls{V} = \clsAtom + \MSet[2]{\cls{V}}\text{.}
\end{align*}
The generating function $V(z)$ for Otter trees satisfies the functional
equation
\begin{align*}
  V(z) = z + \frac{V(z)^2 + V(\cramped{z^2})}{2}\text{,}
\end{align*}
and note that, for this class, we only count external nodes. This
combinatorial class does not have a closed form generating function: prior
Boltzmann samplers for Otter trees have already informally used
approximations~\cite[\S 5]{FlFuPi07}; Pivoteau~\cite{Pivoteau09} used the
fact that $V(z) = 1 - \sqrt{1-2z-V(\cramped{z^2})}$. In practice these
approximations yield correct simulations, but theoretically they could
introduce a bias. With analytic samplers, this possible bias is corrected
by failing with some probability; this also gives us more flexibility to
choose the approximations.

\begin{figure*}[b]
  \begin{bigcenter}
  \raisebox{-0.5\height}{{
      \includegraphics[scale=0.3]{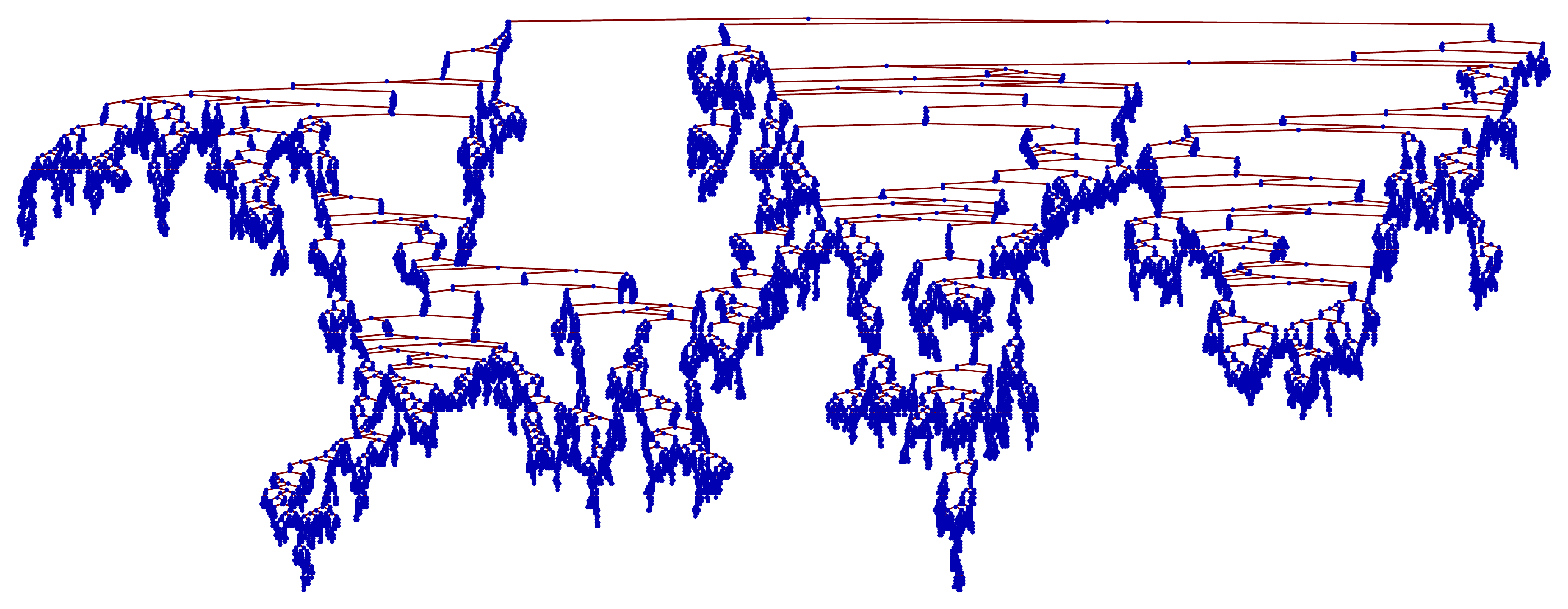}}}\qquad%
  \raisebox{-0.5\height}{{
      \includegraphics[scale=0.4]{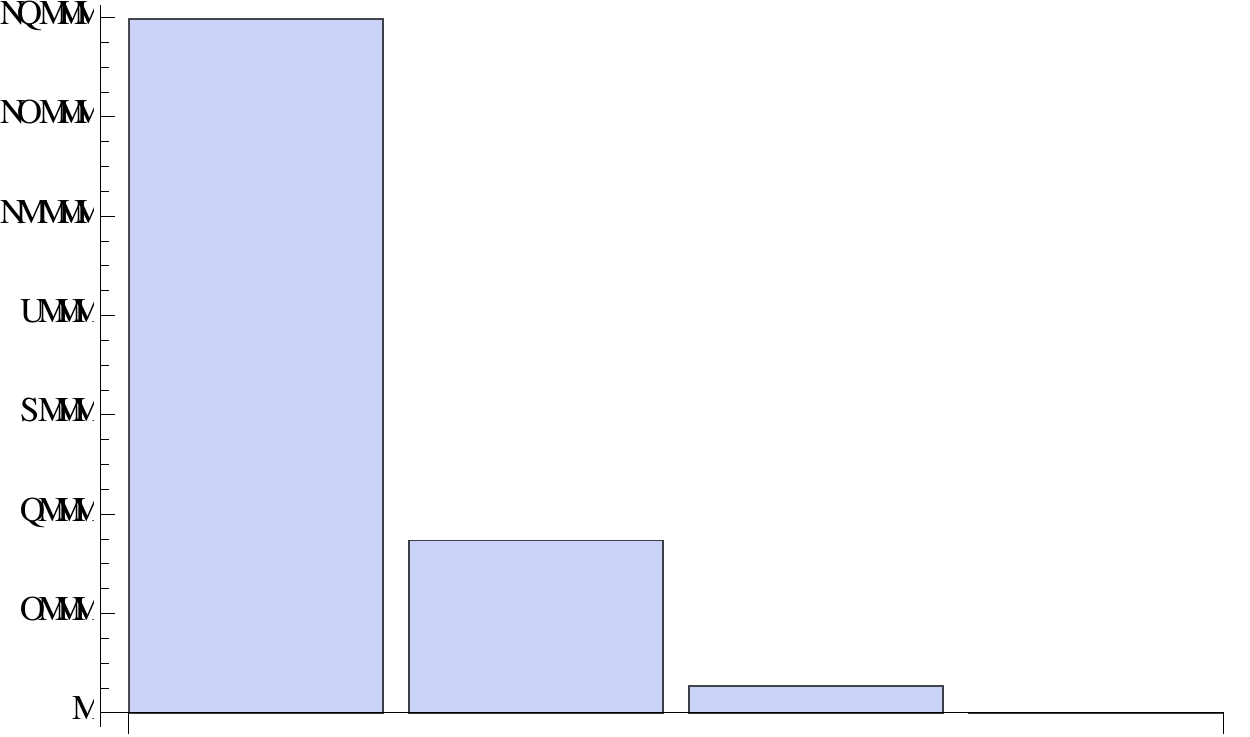}}}
  \end{bigcenter}
  \caption{\label{fig:otter-17979}An Otter tree of size $n=17\,979$ (we
    were targeting $20\,000\pm 10\%$), generated in $13s$ on a standard
    2012 laptop computer. The bar chart summarizes the degree of symmetry
    of the leaves: the first bar indicates how many leaves are not
    duplicated; then duplicated once; then four times; then eight times.}
\end{figure*}

\paragraph{Setting up the inequality} For our analytical samplers, we need
the values $\cramped{v_{[i]}}$, corresponding to $V(\cramped{z^{2^i}})$,
which are defined recursively by the system of inequalities
\begin{align}
  \forall i\in \mathbb{N}^{+}, \label{vi} %
  \quad v_{[i]} \geqslant z^{2^i} + \frac{{v_{[i]}}^2 + v_{[i+1]}}{2}\text{.}
\end{align}
Because this system is infinite, we are first going to pick a threshold
index after which the equations will be approximated; and we will
determine a good approximation for the remaining terms.

In order to find solutions, we need an initial interval for $z$, which
need not be especially precise: to this end, $0 \leqslant z \leqslant 1$
suffices (even though it is simple enough to argue that $1/4 < z < 1/2$).
The constant part of this recursive inequation is $\cramped{z^{2^i}}$,
thus it makes sense to let $\cramped{v_{[i]}}=K\cramped{z^{2^i}}$, which
we can then inject in our inequation. Dividing both sides by
$\cramped{z^{2^i}}$ and factoring, we obtain
\begin{align}\label{eq:otter-k}
  K \geqslant 1 + \frac{Kz^{2^{i}}}{2}\left(K+1\right)\text{.}
\end{align}

\paragraph{Choosing parameters} At this point we now have two parameters
to pick. First we have to find a constant $K$ satisfying
Inequation~\eqref{eq:otter-k}; $K$ can be as small as we want, as long as
$K >1$.

Once we have picked a threshold $\cramped{i_0}$, and the constant $K$
which will approximate terms $\cramped{v_{[i]}}$ for $i > i_0$, we can
exactly compute the initial terms. This is done by solving exactly the
quadratic equations,
\begin{align*}
  \frac 1 2 {v_{[i]}}^2 - v_{[i]} + \left(z^{2^i} + \frac 1 2
    v_{[i+1]}\right) = 0
\end{align*}
going backwards from $\cramped{i_0}-1$ to $0$, and with, as we said, the
remaining terms $\cramped{o_{[i]}} = K \cramped{z^{2^i}}$.

The approximations we have taken here will impact the failure rate, and we
can decrease it by taking any of the following measures: we can pick a
higher threshold $\cramped{i_0}$; we can pick a $z$ that is closer to the
singularity; we can use more than the constant part of the equation in the
step where we reject $\cramped{z^{2^{i+1}}}$ to approximate the terms
beyond the threshold.

This leads to an efficient sampler for Otter trees, of which we have drawn
a very large tree in Figure~\ref{fig:otter-17979}. Consider that this
allows for interesting empirical analyses of these trees.

\section{Conclusion\label{sec:conclusion}}

In this paper, we have proposed to integrate the classical idea of
rejection sampling to the Boltzmann sampler model, therefore relaxing the
condition that generating functions must be evaluated exactly.

The resulting model, which we call \emph{analytic samplers}, is fully
compatible with all prior approaches used in Boltzmann samplers (in
particular, these samplers can work well with Pivoteau~\etal's oracle),
and in fact provides sound theoretical ground by which to allow the
routine approximations that have been made in existing Boltzmann samplers.


But beyond that, we also believe the relaxed properties can allow for
possible improvements and simplifications in the way the samplers target
the size of their output. To illustrate these ideas, we show two types of
applications. First, with the example of simply generated trees, we
illustrate how tuning can be done in an alternate way, by using the added
degree of freedom of exploring points in a region instead of a curve.
Second, we show for Otter trees, that our samplers allow for much larger
approximation to be made with little side-effects.

Some details involve how to simulate the probabilities without resorting
to arbitrary precision: several answers exist, for instance in the form of
Buffon machine as introduced by Flajolet~\etal~\cite{FlPeSo11} or, because
we are not restricted to fixed curve, selecting only rational
probabilities, which can be easily simulated exactly, as shown by
Lumbroso~\cite{Lumbroso13}.

The open question is to determine whether these properties can be
leveraged for large combinatorial systems: indeed, the initial Boltzmann
paper was only illustrated by combinatorial classes defined as one or a
handful of equations. The real impressive strength of the oracle provided
by Pivoteau~\etal~\cite{PiSaSo12} was to be able to handle combinatorial
systems with thousands of equations. It remains to be seen if, when
dealing with much more complex polytopes, it is possible to use simple
refinements of the ideas we have shown for simply generated trees.
Finally, an topic which has not yet reached practical maturity is that of
multidimensional combinatorial classes (where the distribution is not
uniform, but biased according to some combinatorial parameter). The
article of Bodini and Ponty~\cite{BoPo10} has highlighted some issues,
which we believe our present framework might help bypass.

\section*{Acknowledgments}

We would like to acknowledge and thank the considerable feedback we have
received from anonymous referees both on a prior draft of this work, and
on this current submission at ANALCO~2015.

This work was financed in part by the French ANR \textsc{Magnum}. We would
also like to thank the LIA \textsc{Lirco} and J.-C. Aval and J.-F.
Marckert, as well as the ANR \textsc{QuasiCool} for travel funds which
have enabled the second author to give talks on this topic.

\bibliographystyle{plain}
\bibliography{algo,extra,extra2}{}

\end{document}